\documentclass[conference]{IEEEtran}
\IEEEoverridecommandlockouts

\usepackage{cite}
\usepackage{amsmath,amssymb,amsfonts}
\usepackage{algorithmicx}
\usepackage[ruled]{algorithm}
\usepackage{algpseudocode}
\usepackage{graphicx}
\usepackage{textcomp}
\usepackage{xcolor}
\usepackage{float}
\usepackage{tikz}
\usepackage{booktabs}
\usepackage{hyperref}
\usepackage{amsthm}
\usepackage{mdframed}

\newtheorem{definition}{Definition}
\newtheorem{theorem}{Theorem}

\newtheorem{corollary}{Corollary}

\newtheorem{proposition}{Proposition}
\newtheorem{assumption}{Assumption}

\newcommand{\skey}{k}

\newcommand{\noise}{r}              

\def\BibTeX{{\rm B\kern-.05em{\sc i\kern-.025em b}\kern-.08em
    T\kern-.1667em\lower.7ex\hbox{E}\kern-.125emX}}

\IEEEoverridecommandlockouts
\IEEEpubid{\makebox[\columnwidth]{Preprint, 2026 \hfill} \hspace{\columnsep}\makebox[\columnwidth]{ }}

\begin{document}

\title{Beyond Controlled Noise: Achieving Symmetric FHE through Dynamic Position Shifting}

\author{
\IEEEauthorblockN{\small{Mostefa~Kara}}
\IEEEauthorblockA{\textit{\small{LIAP Laboratory, University of El Oued,}} \\
\textit{\small{PO Box 789, El Oued,}}\\
\textit{\small{39000, Algeria}}\\
\small{karamostefa@univ-eloued.dz}}

}

\maketitle
\begin{abstract}
Traditional Fully Homomorphic Encryption (FHE) schemes often suffer from prohibitive computational overhead and complex noise management. In this paper, we propose a novel symmetric FHE through a mechanism of plaintext fragmentation and dynamic interposition. Our approach is built upon a modular encryption foundation, $c = mk + \noise p$, which is naturally additive but typically limited by exponential noise growth during multiplication. To resolve this, we introduce an interposition framework where the plaintext is partitioned into multiple fragments across distinct logical positions. We introduce a dual-regulator system to govern the multiplication process; exponent regulators ($t_i$) redirect the product of fragments to a new target position, preventing the accumulation of secret key exponents, while coefficient regulators ($d_i$) normalize the resulting scalars. Security is established through a binding mechanism where exponents and coefficients are mutually dependent, shielding the secret key $k$ from algebraic manipulation and substitution attacks.
\end{abstract}


\section{Introduction}

Fully Homomorphic Encryption (FHE) has emerged as one of the most powerful paradigms in modern cryptography, enabling computations to be performed directly on encrypted data without revealing the underlying plaintext. Since the seminal work of Gentry \cite{gentry2009fully}, numerous constructions have been proposed to improve the practicality of FHE, ranging from lattice-based designs \cite{brakerski2014leveled, fan2012somewhat} to schemes optimized for machine learning and privacy-preserving cloud computing \cite{chillotti2020tfhe, bergerat2024new}. Despite their theoretical elegance, existing FHE systems remain computationally heavy and resource-intensive, which limits their deployment in lightweight or real-time scenarios such as IoT and WSN \cite{habib2021secure}.  

A promising alternative direction is to explore symmetric FHE constructions, which can offer simpler operations and reduced overhead compared to public-key approaches \cite{kara2020homomorphic}. However, naive symmetric designs suffer from significant limitations \cite{aissaoua2024integrating}. For example, the trivial encryption method defined by $c = m k$, where $m$ is the message and $k$ is the secret key, is an additive homomorphism but rapidly accumulates noise when ciphertexts are multiplied. This uncontrolled growth renders the scheme impractical for applications requiring multiplicative depth.  

To address these challenges, we propose a novel symmetric FHE scheme that combines plaintext fragmentation with an interposition mechanism. The core idea is to divide the message into smaller fragments and use regulator public keys that control the growth of key exponents during homomorphic multiplication. By substituting the positions of plaintext fragments under carefully chosen regulators, the scheme preserves the correctness of multiplicative operations while avoiding noise amplification. This design provides a lightweight and efficient symmetric alternative to conventional FHE models.

The development of FHE has followed several key trajectories. Early works such as Gentry's lattice-based blueprint \cite{gentry2009fully} introduced the concept of bootstrapping to control noise, inspiring follow-up schemes like BGV \cite{brakerski2014leveled} and FV \cite{fan2012somewhat}. These schemes improved efficiency but still relied on complex lattice operations. More recently, TFHE \cite{chillotti2020tfhe} and CKKS \cite{cheon2017homomorphic} tailored FHE to specific domains, supporting fast Boolean circuits or approximate arithmetic, respectively.

In parallel, research on symmetric or lightweight FHE has remained relatively limited. Some studies examined secret-key homomorphic constructions for specialized tasks \cite{armknecht2015guide, catalano2024anamorphic}, but most suffer from scalability or noise-related issues. To the best of our knowledge, no prior work has investigated the use of interposition and plaintext fragmentation to regulate ciphertext growth in symmetric FHE. Our proposal therefore fills this gap by introducing a regulator-assisted mechanism that ensures feasible multiplicative homomorphism in a symmetric setting.

As a result, this paper introduces a novel symmetric FHE scheme that leverages plaintext fragmentation and interposition to enable efficient homomorphic properties (addition, multiplication); formalizes the role of regulator public keys in controlling the exponents of the secret key, thereby addressing the noise growth problem inherent in trivial encryption; analyzes the correctness and security of the scheme, showing how it maintains confidentiality while supporting valid homomorphic operations; and provides an illustrative implementation and example, demonstrating the potential of the scheme as a lightweight alternative to traditional FHE constructions.

\section{Proposal}

In the proposed symmetric FHE scheme, a regulator is a publicly known auxiliary key element designed to govern the evolution of the secret key's exponent during homomorphic multiplication. Formally, a regulator is a value derived from the secret key $k$ and specific positional scalars such that its application mathematically redirects or repositions fragments of the plaintext within the ciphertext domain. The primary purpose of regulators is to ensure that the product of two ciphertexts does not lead to the uncontrolled exponential growth of $k$ or the accumulation of scalar coefficients. By managing these two vectors, the regulators preserve both the correctness of decryption and the structural stability of the ciphertext across recursive operations.

Since providing a direct multiplicative inverse of $k^{x+y}$ would compromise the secret key, the interposition mechanism avoids this by randomly splitting the plaintext into parts and employing a cyclic shift. In a typical multiplication $m_i k^x \times m_i' k^y$, rather than attempting to reduce the result back to the original position $i$ (which would require a risky inverse regulator), we redirect the result to a new target position $j$ using a regulator $k^z$, where $k^{x+y} k^z \neq 1$. This ensures that while the exponent is transformed, the secret $k$ remains algebraically protected. The proposal is naturally probabilistic; even for the same message $m$, the encryption results $Enc_1(m)$ and $Enc_2(m)$ will differ due to the random fragmentation of the plaintext during the initial encoding.

The interposition process can operate in two distinct modes. In the random mode, the product of positions $P_i P_j$ via a regulator is mapped to an arbitrary position $P_l$. In the regular mode, which is adopted in this work with $P=3$ positions, the mapping follows a deterministic function $l = f(i,j)$. Specifically, we define a cyclic flow: if $i=j$, the result moves to $l = i+1 \pmod P$; if $i \neq j$, the result moves to $l = j+1 \pmod P$ (it targets the remaining third position in our case, $P = 3$). This creates a closed-loop system where fragments are perpetually reshuffled but never lost or exponentially expanded.

To achieve total stability, we utilize a dual-regulator architecture consisting of exponent regulators ($t_i$) and coefficient regulators ($d_i$). While the exponent regulators manage the power of $k$, the coefficient regulators normalize the resulting scalar products. Without this normalization, the internal coefficients $a_i$ would square and grow with each multiplication. By applying $d_i$, we perform a coefficient reset that transforms the scalar of the product back into the standard form required for the target position. This dual-binding ensures that an attacker cannot isolate $k$ by manipulating one regulator without being thwarted by the unknown parameters of the other.

The specific interposition mechanism for the three-fragment model ($m_1, m_2, m_3$) is suggested as follows:
\begin{itemize}
    \item $P_1 \times P_1 \times Reg_{t1, d2} \longrightarrow P_2$
    \item $P_2 \times P_2 \times Reg_{t2, d3} \longrightarrow P_3$
    \item $P_3 \times P_3 \times Reg_{t3, d1} \longrightarrow P_1$ (cyclic return)
    \item $P_1 \times P_2 \times Reg_{t4, d3} \longrightarrow P_3$
    \item $P_2 \times P_3 \times Reg_{t6, d1} \longrightarrow P_1$ (cyclic return)
    \item $P_1 \times P_3 \times Reg_{t5, d2} \longrightarrow P_2$
\end{itemize}

\begin{figure}[H]
\includegraphics[width=3in]{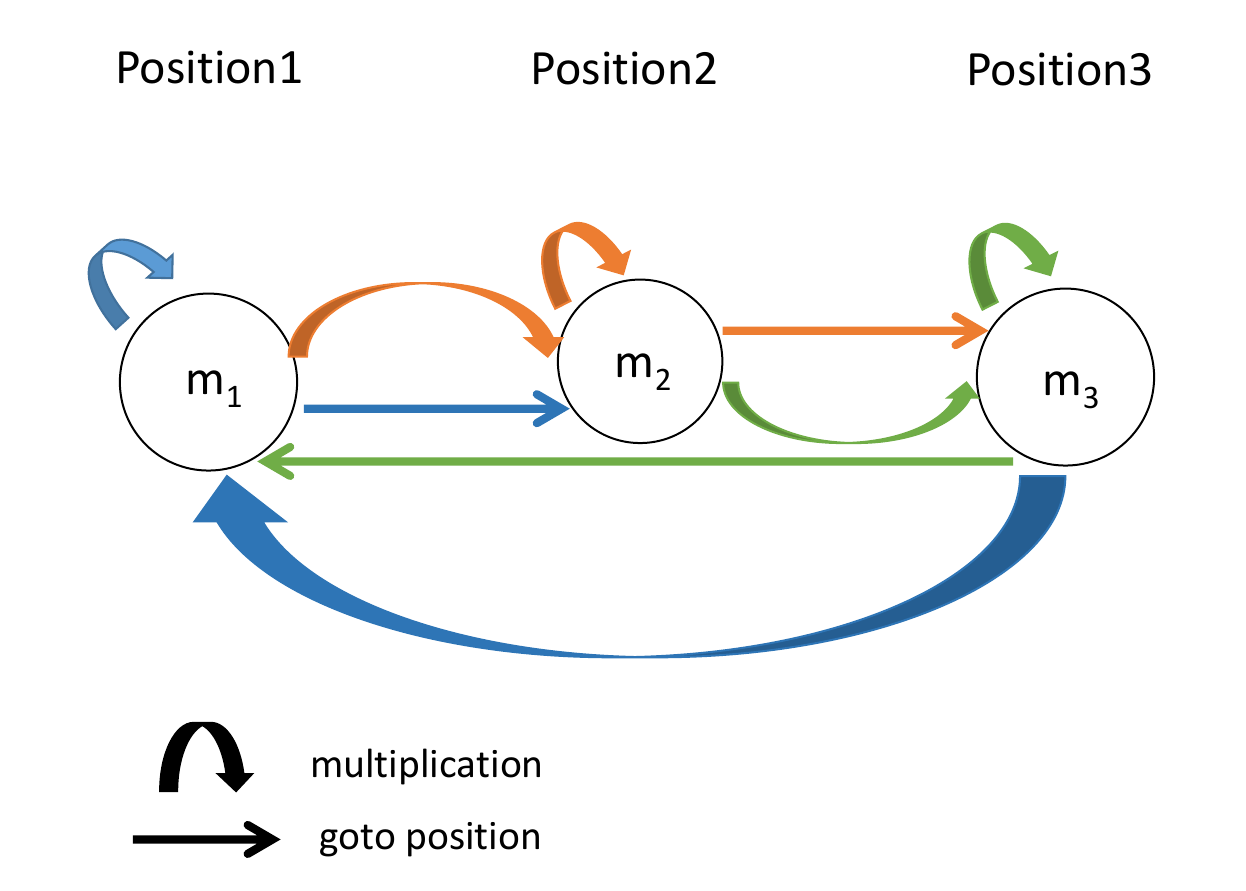}
\caption{A global overview of the interposition mechanism.}
\label{fig:interpos}
\end{figure}

Therefore, six regulators are needed for three fragments (Algorithm \ref{alg:kg}).

\begin{algorithm}
\caption{$\mathsf{KeyGen}(p,\, q,\, R_1,\, R_2)$}
\label{alg:kg}
\begin{algorithmic}[1]
\Require Primes $p$, $q$; sampling bounds $R_1 < R_2$
\Ensure  Secret key $(k_1,k_2,k_3)$;
         public key $(n,\, t_1,\ldots,t_6,\, d_1,d_2,d_3)$

\State $n \gets p \cdot q$

\Statex \Comment{Coefficient scalars}
\State Sample $a_1, a_2, a_3,\, b_1, b_2, b_3 \xleftarrow{\$} [R_1, R_2)$

\State $b_4 \gets a_2 \cdot b_2 \cdot (a_1)^{-1} \pmod{p}$

\State $b_5 \gets a_1 \cdot b_1 \cdot (a_3)^{-1} \pmod{p}$

\State $b_6 \gets a_3 \cdot b_3 \cdot (a_2)^{-1} \pmod{p}$

\Statex \Comment{ Secret key and exponents }
\State Sample $k,\, e_1, e_2, e_3 \xleftarrow{\$} [R_1, R_2)$

\Statex \Comment{ Position keys $k_i = a_i \cdot k^{e_i} + \noise \cdot p \pmod{n}$}
\For{$i \gets 1$ \textbf{to} $3$}
    \State Sample $\noise \xleftarrow{\$} [R_1, R_2)$
    \State $k_i \gets \bigl(a_i \cdot k^{e_i} \bmod p\bigr) + \noise \cdot p \pmod{n}$
\EndFor

\Statex \Comment{Exponent regulators $t_i = b_i \cdot k^{x} + \noise \cdot p \pmod{n}$}
\State Sample $\noise$;
\quad $t_1 \gets b_1 \cdot k^{(e_2 - 2e_1)}  + \noise \cdot p \pmod{n}$

\State Sample $\noise$;
\quad $t_2 \gets b_2 \cdot k^{(e_3 - 2e_2)} + \noise \cdot p \pmod{n}$

\State Sample $\noise$;
\quad $t_3 \gets b_3 \cdot k^{(e_1 - 2e_3)} + \noise \cdot p \pmod{n}$

\State Sample $\noise$;
\quad $t_4 \gets b_4 \cdot k^{(e_3 - e_1 - e_2)}  + \noise \cdot p \pmod{n}$

\State Sample $\noise$;
\quad $t_5 \gets b_5 \cdot k^{(e_2 - e_1 - e_3)} +  \noise \cdot p \pmod{n}$

\State Sample $\noise$;
\quad $t_6 \gets b_6 \cdot k^{(e_1 - e_2 - e_3)} + \noise \cdot p \pmod{n}$

\Statex \Comment{ Coefficient regulators }
\State $d_1 \gets a_1 \cdot (a_3^2 \cdot b_3)^{-1} \pmod{p}$

\State $d_2 \gets a_2 \cdot (a_1^2 \cdot b_1)^{-1} \pmod{p}$

\State $d_3 \gets a_3 \cdot (a_2^2 \cdot b_2)^{-1} \pmod{p}$

\State \Return $(k_1, k_2, k_3,\; t_1,\ldots,t_6,\; d_1, d_2, d_3,\; n)$
\end{algorithmic}
\end{algorithm}

\begin{algorithm}
\caption{$\mathsf{Enc}(m,\; k_1, k_2, k_3,\; p,\, n,\; R_1, R_2)$}
\begin{algorithmic}[1]
\Require Plaintext $m$; position keys $k_1,k_2,k_3$; moduli $p$, $n$;
         bounds $R_1, R_2$
\Ensure  Ciphertext $(c_1, c_2, c_3)$

\State Sample $m_1, m_2 \xleftarrow{\$} [R_1, R_2)$
\State $m_3 \gets (m - m_1 - m_2) \bmod n$
\Comment{Ensures $m_1 + m_2 + m_3 = m$; encoding is probabilistic}

\For{$i \gets 1$ \textbf{to} $3$}
    \State Sample $\noise \xleftarrow{\$} [R_1, R_2)$
    \State $c_i \gets m_i \cdot k_i + \noise \cdot p \pmod{n}$
    \Comment{$c_i \equiv m_i \cdot k_i \pmod{p}$; noise vanishes mod $p$}
\EndFor

\State \Return $(c_1, c_2, c_3)$
\end{algorithmic}
\end{algorithm}

\begin{algorithm}
\caption{$\mathsf{Dec}(c_1, c_2, c_3,\; k_1^{-1}, k_2^{-1}, k_3^{-1},\; p)$}
\begin{algorithmic}[1]
\Require Ciphertext $(c_1, c_2, c_3)$;
         modular inverses $k_i^{-1} = (k_i \bmod p)^{-1} \pmod{p}$;
         prime $p$
\Ensure  Plaintext $m$

\For{$i \gets 1$ \textbf{to} $3$}
    \State $m_i \gets c_i \cdot k_i^{-1} \pmod{p}$
    \Comment{Noise term $\noise \cdot p \equiv 0 \pmod{p}$, so it vanishes}
\EndFor
\State $m \gets (m_1 + m_2 + m_3) \pmod{p}$

\State \Return $m$
\end{algorithmic}
\end{algorithm}

\begin{algorithm}
\caption{$\mathsf{HAdd}(c_1,c_2,c_3,\; c_1',c_2',c_3',\; n)$}
\begin{algorithmic}[1]
\Require Two ciphertexts $(c_1,c_2,c_3)$, $(c_1',c_2',c_3')$; modulus $n$
\Ensure  Ciphertext $(c_1'',c_2'',c_3'')$ encrypting $m + m'$

\For{$i \gets 1$ \textbf{to} $3$}
    \State $c_i'' \gets (c_i + c_i') \bmod n$
    \Comment{$(m_i + m_i') \cdot k_i \pmod{p}$; position and exponent unchanged}
\EndFor

\State \Return $(c_1'', c_2'', c_3'')$
\end{algorithmic}
\end{algorithm}

\begin{algorithm}
\caption{$\mathsf{HMul}(c_1,c_2,c_3,\; c_1',c_2',c_3',\;
                        t_1,\ldots,t_6,\; d_1,d_2,d_3,\; n)$}
\begin{algorithmic}[1]
\Require Two ciphertexts $(c_1,c_2,c_3)$, $(c_1',c_2',c_3')$;
         exponent regulators $t_1,\ldots,t_6$;
         coefficient regulators $d_1,d_2,d_3$; modulus $n$
\Ensure  Ciphertext $(c_1'',c_2'',c_3'')$ encrypting $m \cdot m'$

\Statex \Comment{%
Each output position accumulates the products whose interposition map targets it. $d_i$ normalises the resulting coefficient so decryption yields the correct fragment.}

\State $c_1'' \gets \bigl(c_3 c_3' t_3 + c_2 c_3' t_6 + c_3 c_2' t_6\bigr)\cdot d_1 \pmod{n}$
\Comment{%
Targets $P_1$: same-pos $P_3\!\times\!P_3$ via $t_3$;
cross-pos $P_2\!\times\!P_3$ (both orders) via $t_6$}

\State $c_2'' \gets \bigl(c_1 c_1' t_1 + c_1 c_3' t_5 + c_3 c_1' t_5\bigr)\cdot d_2 \pmod{n}$
\Comment{%
Targets $P_2$: same-pos $P_1\!\times\!P_1$ via $t_1$;
cross-pos $P_1\!\times\!P_3$ (both orders) via $t_5$}

\State $c_3'' \gets \bigl(c_2 c_2' t_2 + c_1 c_2' t_4 + c_2 c_1' t_4\bigr)\cdot d_3 \pmod{n}$
\Comment{%
Targets $P_3$: same-pos $P_2\!\times\!P_2$ via $t_2$;
cross-pos $P_1\!\times\!P_2$ (both orders) via $t_4$}

\State \Return $(c_1'', c_2'', c_3'')$
\end{algorithmic}
\end{algorithm}

\section{Security Analysis}


\newcommand{\Adv}{\mathcal{A}}
\newcommand{\Enc}{\mathsf{Enc}}
\newcommand{\Dec}{\mathsf{Dec}}
\newcommand{\KeyGen}{\mathsf{KeyGen}}
\newcommand{\HMul}{\mathsf{HMul}}
\newcommand{\HAdd}{\mathsf{HAdd}}
\newcommand{\negl}{\mathsf{negl}}
\newcommand{\poly}{\mathsf{poly}}
\newcommand{\Zn}{\mathbb{Z}_n}
\newcommand{\Zp}{\mathbb{Z}_p}
\newcommand{\AdvIND}[1]{\mathbf{Adv}^{\mathrm{IND\text{-}CPA}}_{#1}}
\newcommand{\abs}[1]{\left|#1\right|}
\newcommand{\MMA}{\mathsf{MMA}}
\newcommand{\HMP}{\mathsf{HMP}}

\subsection{Preliminaries and Notation}
 
Let $\lambda$ denote the security parameter. We write $\negl(\lambda)$ for any negligible function and $\poly(\lambda)$ for an unspecified polynomial. All algorithms are PPT unless stated otherwise.
 
\paragraph{Scheme parameters}
\begin{itemize}
  \item $p, q$: distinct primes of bit-length $\approx \lambda$;\quad $n = pq$.
  \item $[R_1, R_2)$: deployment-defined sampling range with $\Delta := R_2 - R_1 \geq 2^\lambda$ (super-polynomial in $\lambda$).
  \item Plaintext space: $m \in \Zp$; fragments satisfy $m_1 + m_2 + m_3 \equiv m \pmod{p}$ with $m_1, m_2 \xleftarrow{\$} \Zp^*$, $m_3 := m - m_1 - m_2 \bmod p$.
  \item Secret key $\skey$, exponents $e_i$, scalars $a_i, b_i$: all sampled from $[R_1, R_2)$.
  \item Position keys: $k_i \equiv a_i \cdot \skey^{e_i} \pmod{p}$, embedded in $\Zn$ with masking noise $r_i p$.
  \item Ciphertext of fragment $m_i$:
        $c_i = m_i k_i + r_i p \pmod{n}$, so
        $c_i \equiv m_i k_i \pmod{p}$.
\end{itemize}
 
\subsection{Hardness Assumptions}
 
We identify the precise assumptions on which the IND-CPA proof rests. 
 
\begin{assumption}[Multiplicative Masking Assumption ($\MMA_p$)]
\label{asm:mma}
Let $p$ be a $\lambda$-bit prime and $k_i \in \Zp^*$ be a fixed (secret) group element. For a uniformly random $u \xleftarrow{\$} \Zp^*$, the distribution of $u \cdot k_i \bmod p$ is computationally indistinguishable from the uniform distribution on $\Zp^*$ for any PPT adversary that does not know $k_i$. Formally, for all PPT $\Adv$:
\[
  \Bigl|
    \Pr\bigl[\Adv(p,\, u \cdot k_i \bmod p) = 1\bigr]
    -
    \Pr\bigl[\Adv(p,\, v) = 1 \mid v \xleftarrow{\$} \Zp^*\bigr]
  \Bigr| = \negl(\lambda).
\]
\end{assumption}
 
$\MMA_p$ holds unconditionally when $k_i$ is secret and $u$ is uniform: multiplication by a fixed nonzero element is a bijection on $\Zp^*$, so $u \cdot k_i$ is identically distributed to $u$ when $k_i$ is unknown and $u$ is uniform. The assumption, therefore, reduces to the secrecy of $k_i$, which in turn rests on DLA (recovering $k_i$ from public information requires solving a discrete logarithm). We state it as a named assumption for modularity, but note that once $k_i$ is secret, the masking is information-theoretically perfect on $\Zp^*$.
 
\begin{assumption}[Hidden Modulus Indistinguishability ($\HMP_n$)]
\label{asm:hmp}
Let $n = pq$ with $p, q$ unknown $\lambda$-bit primes. For a uniformly random $x \xleftarrow{\$} \Zp^*$ and uniformly random $r \xleftarrow{\$} [R_1, R_2)$, the distribution
\[
  c = x + r \cdot p \pmod{n}
\]
is computationally indistinguishable from the uniform distribution on $\Zn$ for any PPT adversary that does not know $p$. Formally, for all PPT $\Adv$:
\[
  \Bigl|
    \Pr\bigl[\Adv(n,\, c) = 1\bigr]
    -
    \Pr\bigl[\Adv(n,\, w) = 1 \mid w \xleftarrow{\$} \Zn\bigr]
  \Bigr| = \negl(\lambda).
\]
\end{assumption}
 
 
\subsection{IND-CPA Security}
 
\subsubsection{Security Game}
 
\begin{definition}[IND-CPA Game]
The experiment $\mathbf{Exp}^{\mathrm{IND\text{-}CPA}}_{\Pi,\Adv}(\lambda)$:
\begin{enumerate}
  \item Setup. Run $\KeyGen(\lambda)$ to obtain $\mathit{sk} = (k_1,k_2,k_3)$ and $\mathit{pk} = (n,\, t_1,\ldots,t_6,\, d_1,d_2,d_3)$.
  \item Learning. $\Adv(\mathit{pk})$ queries $\Enc(\mathit{sk},\cdot)$ polynomially many times.
  \item Challenge. $\Adv$ outputs $(m_0, m_1)$. Challenger samples $b \xleftarrow{\$} \{0,1\}$, returns $\mathbf{c}^* \leftarrow \Enc(\mathit{sk}, m_b)$.
  \item Guess. $\Adv$ outputs $b'$.
\end{enumerate}
$\AdvIND{\Pi}(\Adv) := \bigl|\Pr[b'=b] - \tfrac{1}{2}\bigr|.$
\end{definition}
 
\subsubsection{Main Theorem}
 
\begin{theorem}[IND-CPA Security under $\MMA_p$ and $\HMP_n$]
\label{thm:indcpa}
For every PPT adversary $\Adv$,
\[
  \AdvIND{\Pi}(\Adv)
  \;\leq\;
  3\cdot\mathbf{Adv}^{\MMA}(\Adv')
  \;+\;
  \mathbf{Adv}^{\HMP}(\Adv'')
\]
where $\Adv', \Adv''$ are PPT algorithms derived from $\Adv$. In particular, if $\MMA_p$ and $\HMP_n$ hold, then $\AdvIND{\Pi}(\Adv) = \negl(\lambda)$.
\end{theorem}
 
\begin{proof}
We proceed via four hybrid games. Let $\Pr_i$ denote the probability that $\Adv$ outputs $b' = b$ in game $G_i$.
 
\paragraph{Game $G_0$} Real experiment.\\
The challenge ciphertext is $\mathbf{c}^* = \Enc(\mathit{sk}, m_b)$:
\begin{align*}
  m_1^*, m_2^* &\xleftarrow{\$} \Zp^*, \quad
  m_3^* := m_b - m_1^* - m_2^* \bmod p,\\
  c_i^* &= m_i^* k_i + r_i p \pmod{n}, \quad r_i \xleftarrow{\$} [R_1,R_2).
\end{align*}
 
\paragraph{Game $G_1$} Re-randomize fragments, preserving the sum. \\
Replace the fragment triple with a freshly sampled one, subject to the same sum constraint:
\[
  u_1, u_2 \xleftarrow{\$} \Zp^*, \quad u_3 := m_b - u_1 - u_2 \bmod p.
\]
Compute $c_i^{**} = u_i k_i + r_i' p \pmod{n}$, $r_i' \xleftarrow{\$} [R_1,R_2)$. Transition $G_0 \to G_1$ is exact. The distribution of $(m_1^*, m_2^*, m_3^*)$ in $G_0$ and of $(u_1, u_2, u_3)$ in $G_1$ are identical both are uniform over the affine subspace $\{(x_1,x_2,x_3) \in (\Zp^*)^3 : x_1+x_2+x_3 = m_b\}$. The noises $r_i, r_i'$ are also identically distributed. Hence $G_0$ and $G_1$ are perfectly indistinguishable:
\[
  \abs{\Pr_0 - \Pr_1} = 0.
\]

\paragraph{Game $G_2$} Replace $u_i k_i \bmod p$ with uniform elements.\\
For each $i = 1,2,3$, replace $u_i k_i \bmod p$ with an independently uniform $v_i \xleftarrow{\$} \Zp^*$, producing $c_i^{***} = v_i + r_i'' p \pmod{n}$.
 
Transition $G_1 \to G_2$ reduces to $\MMA_p$. Fix $i \in \{1,2,3\}$. In $G_1$, the value $u_i k_i \bmod p$ is the product of a uniformly random $u_i \in \Zp^*$ with the secret group element $k_i \in \Zp^*$. By $\MMA_p$ (Assumption~\ref{asm:mma}), this product is computationally indistinguishable from a uniform element of $\Zp^*$ for any PPT adversary that does not know $k_i$.
 
Formally, suppose $\Adv$ distinguishes $G_1$ from $G_2$ with advantage $\epsilon_i$ on component $i$. We build a PPT algorithm $\Adv'$ that receives a $\MMA_p$ challenge $(p, w)$, where $w$ is either $u \cdot k_i$ for uniform $u$, or uniform $v$, and simulates the remaining components honestly. Then $\Adv'$ breaks $\MMA_p$ with the same advantage $\epsilon_i$.
 
A union bound over the three components gives:
\[
  \abs{\Pr_1 - \Pr_2} \leq 3 \cdot \mathbf{Adv}^{\MMA}(\Adv').
\]
 
\paragraph{Game $G_3$} Refresh the noise.\\
Re-sample $r_i'' \xleftarrow{\$} [R_1, R_2)$ independently. Since $r_i''$ is already uniform and independent of $v_i$, the distribution of $c_i^{***}$ is unchanged:
\[
  \abs{\Pr_2 - \Pr_3} = 0.
\]
 
\paragraph{Game $G_4$} Replace ciphertext components with uniform elements of $\Zn$.\\
Replace each $c_i^{***} = v_i + r_i'' p \pmod{n}$ with a uniformly random $w_i \xleftarrow{\$} \Zn$.
 
Transition $G_3 \to G_4$ reduces to $\HMP_n$. In $G_3$, each $c_i^{***}$ has the form $v_i + r p \bmod n$ with $v_i \xleftarrow{\$} \Zp^*$ and $r \xleftarrow{\$} [R_1,R_2)$. This is precisely the distribution considered in Assumption~\ref{asm:hmp}. A distinguisher $\Adv$ that separates $G_3$ from $G_4$ on component $i$ directly yields a PPT algorithm $\Adv''$ breaking $\HMP_n$, because $\Adv''$ can forward the $\HMP_n$ challenge $c$ as the $i$-th ciphertext component and simulate the rest honestly.
 
A union bound over three components gives:
\[
  \abs{\Pr_3 - \Pr_4} \leq 3 \cdot \mathbf{Adv}^{\HMP}(\Adv'').
\]
The factor $3$ is absorbed into the $\mathbf{Adv}^{\HMP}$ term by a standard hybrid over components.
 
Therefore, in $G_4$, each ciphertext component is uniform and independent in $\Zn$, carrying no information about $b$. Hence $\Pr_4 = 1/2$. Chaining:
\begin{align*}
  \AdvIND{\Pi}(\Adv)
  &= \abs{\Pr_0 - \tfrac{1}{2}} \\
\end{align*} 
\begin{align*}
  &\leq \abs{\Pr_0 - \Pr_1}
      + \abs{\Pr_1 - \Pr_2}
      + \abs{\Pr_2 - \Pr_3}
      + \abs{\Pr_3 - \Pr_4}
      + \abs{\Pr_4 - \tfrac{1}{2}} \\
\end{align*}
\begin{align*}  
  &\leq 0
      + 3\cdot\mathbf{Adv}^{\MMA}(\Adv')
      + 0
      + 3\cdot\mathbf{Adv}^{\HMP}(\Adv'')
      + 0. \qquad\square
\end{align*}
\end{proof}
 
\subsubsection{Probabilistic Encryption and Ciphertext Unlinkability}
 
\begin{corollary}[Ciphertext Unlinkability]
For any fixed $m$ and two independent calls $\mathbf{c}_1 \leftarrow \Enc(\mathit{sk},m)$, $\mathbf{c}_2 \leftarrow \Enc(\mathit{sk},m)$, no PPT adversary can distinguish the pair $(\mathbf{c}_1,\mathbf{c}_2)$ from two independent encryptions of different messages, except with advantage $\negl(\lambda)$.
\end{corollary}
 
\begin{proof}
Any such distinguisher $\Adv$ yields an IND-CPA adversary $\Adv'$: on receiving challenge $\mathbf{c}^*$, $\Adv'$ generates a fresh $\mathbf{c}_2 \leftarrow \Enc(\mathit{sk}, m_0)$ and feeds $(\mathbf{c}^*, \mathbf{c}_2)$ to $\Adv$. If $\Adv$ can link $\mathbf{c}^*$ to $m_0$ or $m_1$, $\Adv'$ recovers $b$, breaking IND-CPA. The advantage of $\Adv'$ equals that of $\Adv$.
\end{proof}
 
\subsection{Key Recovery Analysis}
\label{sec:keyrecovery}
 
The IND-CPA proof above treats $k_i$ as an opaque secret group element. This part analyses how hard it is to recover $k_i$, $\skey$, $e_i$, or $a_i$ from the public key, grounding secrecy of $k_i$ in DLA.
 
\begin{proposition}[Hardness of Position Key Recovery]
Any PPT adversary that recovers $k_i \bmod p$ from the secret key with non-negligible probability solves DLA in $\Zp^*$.
\end{proposition}
 
\begin{proof}
The secret key contains $k_i \bmod n$, not $k_i \bmod p$ directly; extracting $k_i \bmod p$ requires computing $k_i \bmod p = (k_i \bmod n) \bmod p$, which requires knowing $p$, hence factoring $n$ (IFA). Assuming $p$ is known (or factored), recovering the base $\skey$ from $k_i \equiv a_i \skey^{e_i} \pmod p$ requires computing a discrete logarithm with an unknown coefficient $a_i$,  at least as hard as DLA.
\end{proof}
 
\begin{proposition}[Exponent Hiding by Regulators]
The exponent regulators $\{t_i\}$ do not leak the exponents $\{e_i\}$ under DLA in $\Zp^*$.
\end{proposition}
 
\begin{proof}
Each regulator satisfies $t_i \equiv b_j \cdot \skey^x \pmod p$, where $x$ is a private linear combination of $e_1,e_2,e_3$, and $b_j$ is a secret scalar. Recovering $x$ requires solving DLA with an unknown coefficient, strictly harder than standard DLA.
\end{proof}
 
\begin{proposition}[Dual-Binding Security]
No PPT adversary can forge a valid pair $(t_i^*, d_i^*)$ with $(t_i^*, d_i^*) \neq (t_i, d_i)$ that produces correct decryption after homomorphic multiplication, without solving DLA.
\end{proposition}
 
\begin{proof}
Forging $t_i^*$ to pass decryption requires inducing the same exponent offset on $\skey$, i.e.\ $\skey^{x^*} = \skey^x$, forcing $x^* = x$ in $\Zp^*$, a DLA instance. Independently forging $d_i^*$ requires knowledge of the secret scalars $(a_j, b_j)$ embedded inside $k_j$ and $t_j$, whose extraction also reduces to DLA. The two constraints are independent, so simultaneous forgery requires solving DLA twice.
\end{proof}
 

Dual binding is a defense against exponent cancellation. Without $d_i$, an adversary who can set $t_i = \skey^{-(e_a+e_b)}$ could cancel the exponent entirely after multiplication, reducing the product ciphertext to $m_a m_b \cdot a_a a_b$, a purely algebraic (non-hidden) quantity. The coefficient regulator $d_i = a_i(a_j^2 b_j)^{-1} \bmod p$ prevents this; correct normalization requires $a_j$ and $b_j$, both secrets.
 
\subsection{IND-CCA Insecurity}
 
\begin{corollary}[The Scheme is NOT IND-CCA Secure]
$\Pi$ does not achieve IND-CCA2 security.
\end{corollary}
 
\begin{proof}
Given challenge $\mathbf{c}^* = \Enc(\mathit{sk}, m_b)$, the adversary computes $\mathbf{c}^{**} = \HAdd(\mathbf{c}^*, \Enc(\mathit{sk}, 0))$ (a differently distributed ciphertext), queries the decryption oracle on $\mathbf{c}^{**}$ (which is not $\mathbf{c}^*$), recovers $m_b$, and wins with advantage $1$.
\end{proof}
 
This is the expected behavior for any homomorphic scheme; the homomorphic property is a form of controlled malleability, which is definitionally incompatible with CCA security. All standard FHE schemes (BGV, BFV, CKKS, GSW) are IND-CPA but not IND-CCA. If CCA security is required in a specific application, a standard CPA-to-CCA transform (e.g., Fujisaki-Okamoto) can be applied at the application layer.
 
\subsection{Homomorphic Correctness and Plaintext Capacity}
 
\subsubsection{Correctness of Addition}
 
After $\HAdd$, each component satisfies $c_i'' \equiv (m_i + m_i') k_i \pmod p$. Summing the recovered fragments gives $m + m'$. Addition is exact in $\Zp$; the additive noise $r_i p$ vanishes upon reduction mod $p$.
 
\subsubsection{Correctness of Multiplication}
 
Working mod $p$, the output of $\HMul$ at position $P_1$ is verified as:
\[
  c_3 c_3' t_3 \equiv m_3 m_3' \cdot a_3^2 b_3 \cdot \skey^{e_1},
\]
and, using $b_6 = a_3 b_3 a_2^{-1} \bmod p$:
\[
  c_2 c_3' t_6 \equiv m_2 m_3' \cdot a_3^2 b_3 \cdot \skey^{e_1},
  \qquad
  c_3 c_2' t_6 \equiv m_3 m_2' \cdot a_3^2 b_3 \cdot \skey^{e_1}.
\]
Applying $d_1 = a_1(a_3^2 b_3)^{-1} \bmod p$:
\[
  c_1'' \equiv a_1 \skey^{e_1} \cdot (m_3 m_3' + m_2 m_3' + m_3 m_2').
\]
Summing $c_1'' + c_2'' + c_3''$ across all positions recovers $m \cdot m'$ exactly. Positions $P_2$ and $P_3$ follow by the same calculation with indices permuted.
 
\subsubsection{Plaintext Capacity and Multiplicative Depth}
 
\begin{proposition}[Unbounded Multiplicative Depth up to Plaintext Capacity]
The scheme correctly evaluates any arithmetic circuit of multiplication depth $L$, for any $L \geq 1$, subject only to
\[
  \prod_{j=1}^{L} m^{(j)} \;<\; p.
\]
\end{proposition}
 
\begin{proof}
All homomorphic operations are exact in $\Zp$: noise vanishes mod $p$ at every step, and there is no stochastic error term. The result of $L$ multiplications is the exact integer product $M = \prod_j m^{(j)}$. Decryption recovers $M \bmod p = M$ whenever $M < p$. No other constraint on $L$ exists.
\end{proof}
 
Unlike LWE-based schemes, there is no noise budget, no modulus switching, and no bootstrapping required. The constraint $\prod_j m^{(j)} < p$ is deterministic and fully predictable from the plaintext values. It is controlled by (i) choosing $p$ large relative to the expected product, and (ii) restricting the plaintext domain to $m \in [1,B]$ with $B^L < p$ for depth-$L$ circuits. The sampling range $[R_1, R_2)$ for key material affects only the statistical security of the fragment distribution; it has no influence on multiplicative depth or correctness.
 
\subsection{Resistance to Specific Attacks}
 
\subsubsection{Known-Plaintext Attack}
Given $(m^{(j)}, \mathbf{c}^{(j)})$: recovering $k_i$ from $c_i^{(j)} = m_i^{(j)} k_i + r_i^{(j)} p \bmod n$ requires knowing $m_i^{(j)}$ individually, not just the total $m^{(j)}$, which requires solving an underdetermined system (2 unknowns, 1 equation per ciphertext). Even with $m_i^{(j)}$, recovering $k_i$ from $\Zn$ requires $\HMP_n$, and recovering $\skey$ from $k_i$ requires DLA.

\begin{table*}[h!]
\centering
\caption{Summary Table}
\label{tab:Summary}
\begin{tabular}{|l|c|c|}
\hline
\textbf{Property} & \textbf{Status} & \textbf{Basis} \\
\midrule
IND-CPA
  & \textcolor{green!60!black}{\textbf{Yes}}
  & $\MMA_p$ + $\HMP_n$ \\
IND-CCA
  & \textcolor{red}{\textbf{No}}
  & Inherent; homomorphic malleability \\
Ciphertext unlinkability
  & \textcolor{green!60!black}{\textbf{Yes}}
  & IND-CPA reduction \\
Position key secrecy
  & \textcolor{green!60!black}{\textbf{Yes}}
  & DLA + $\HMP_n$ \\
Exponent hiding ($t_i$)
  & \textcolor{green!60!black}{\textbf{Yes}}
  & DLA \\
Coefficient hiding ($d_i$)
  & \textcolor{green!60!black}{\textbf{Yes}}
  & DLA \\
Dual-binding (joint $t_i, d_i$)
  & \textcolor{green!60!black}{\textbf{Yes}}
  & DLA \\
Resistance to KPA
  & \textcolor{green!60!black}{\textbf{Yes}}
  & Underdetermined + DLA + $\HMP_n$ \\
Resistance to linear fragment attack
  & \textcolor{green!60!black}{\textbf{Yes}}
  & Underdetermined for all $L$ \\
Exact homomorphic addition
  & \textcolor{green!60!black}{\textbf{Yes}}
  & Exact in $\Zp$ \\
Exact homomorphic multiplication
  & \textcolor{green!60!black}{\textbf{Yes}}
  & Algebraically verified \\
Unbounded multiplicative depth
  & \textcolor{green!60!black}{\textbf{Yes}$^*$}
  & $^*$Subject to $\prod m^{(j)} < p$ \\
No noise accumulation
  & \textcolor{green!60!black}{\textbf{Yes}}
  & Noise vanishes mod $p$ \\
Post-quantum security
  & \textcolor{red}{\textbf{No}}
  & Shor breaks DLA and IFA \\
\midrule
\multicolumn{3}{l}{%
  \textbf{Proof rests on:}\quad
  $\MMA_p$ (masking) $+$ $\HMP_n$ (hidden modulus)
  $\Leftarrow$ DLA $+$ IFA (classical).
}\\
\bottomrule
\hline
\end{tabular}
\end{table*}

\subsubsection{Linear Algebraic Attack on Fragments}
With $L$ ciphertexts of $m$: $c_i^{(\ell)} \equiv m_i^{(\ell)} k_i \pmod p$ gives $L$ equations per position in $2L$ unknowns $(m_1^{(\ell)},m_2^{(\ell)})$. The system is underdetermined for all finite $L$, so the attack fails regardless of the number of observed ciphertexts.
 
\subsubsection{Regulator Manipulation}
Captured formally by the Dual-Binding proposition (Section~\ref{sec:keyrecovery}).
 
\subsubsection{Chinese Remainder Theorem (CRT) / Hidden Modulus Attack}
$\Zn \cong \Zp \times \mathbb{Z}_q$ by CRT. Separating the two components requires knowing $p$, i.e.\ solving $\HMP_n$. Without $p$, the coset structure $c_i \in m_i k_i + p\Zn$ is computationally hidden.
 
\subsubsection{Quantum Attacks}
The scheme relies on the secrecy of $k_i$, which is protected by DLA and IFA/HMP. Shor's algorithm solves DLA and IFA in polynomial quantum time, breaking both $\MMA_p$ (by revealing $k_i$) and $\HMP_n$ (by factoring $n$). Thus, the scheme is not post-quantum secure (Table \ref{tab:Summary}).


\section{Performance}

Information about the benchmarked schemes is taken from \cite{doan2023survey,jorge2025evaluating,acar2018survey}. For BGV and BFV, the parameters are set to $N=8192$ and $\log(q)=218$. For our proposal, the modulus size is $n=3072$ bits, corresponding to the 128-bit classical security level. The implementation was carried out on an Intel(R) Core(TM) i7-10700 CPU running at 2.90~GHz under Windows.

\begin{table*}[h!]
\centering
\caption{Performance and Feature Comparison of Selected FHE Schemes (ms)}
\label{tab:comparison}
\begin{tabular}{|l|c|c|c|c|c|c|c|c|}
\hline
\textbf{Scheme / Feature} & \textbf{KeyGen} & \textbf{Enc} & \textbf{Dec} & \textbf{Add} & \textbf{Mlt} & \textbf{Cipher- size} & \textbf{Noise Control} & \textbf{Limitation} \\
\hline
YASHE & NA    & 16    & 15    & 0.7    & 18      & NA      & Bootstrapping     & Computationally expensive \\
BFV  & 3.003 & 3.269 & 1.179 & 0.144  & 11.66   & 446 KB  & Relinearization    & Large ciphertext size \\
BGV  & 11.42 & 3.137 & 0.992 & 0.079  & 6.673   & 446 KB  & Bootstrapping      & Large ciphertext size \\
TFHE  & NA    & 29.1  & 1.8   & \multicolumn{2}{c|}{2,308,697} & 9.7 KB  & Bootstrapping      & Inefficient for extensive arithmetic \\
CKKS  & NA    & 3344  & 1182  & \multicolumn{2}{c|}{915.1}   & 500 KB  & Rescaling, Precision loss & Inexact results; requires precision management \\
\hline
\textbf{Proposal} & 20.3  & 0.02 & 0.051 & 0.002 & 0.22  & 9 KB   & Interposition / Regulator & Symmetric; quantum vulnerability \\
\hline
\end{tabular}
\end{table*}

Table~\ref{tab:comparison} summarizes the performance and features of several well-established FHE schemes compared to our proposal. Existing lattice-based schemes such as BFV and BGV provide strong asymptotic security but suffer from large ciphertext sizes (hundreds of KB) and require costly bootstrapping or relinearization for noise management. CKKS supports approximate arithmetic but incurs significant overhead and suffers from precision loss, while TFHE is extremely fast for single-bit operations but inefficient for large-scale arithmetic. 

In contrast, the proposed scheme demonstrates lightweight performance. The encryption time is approximately $0.02$ ms, and addition is nearly negligible ($0.002$ ms), with a ciphertext size of only $9$ KB. Noise is controlled efficiently through the regulator-based interposition mechanism rather than costly bootstrapping. The main limitation is that the scheme is symmetric, and its security ultimately depends on the hardness of integer factorization, implying potential vulnerability in the post-quantum setting. Nevertheless, within the classical model it provides a highly efficient alternative for applications that prioritize speed and compactness.

\section{Conclusion}
In this paper, we proposed a novel symmetric fully homomorphic encryption scheme that leverages plaintext fragmentation and an interposition mechanism based on regulator values to enable efficient homomorphic addition and multiplication. The design builds upon the lightweight trivial encryption technique while overcoming its main limitation of uncontrolled noise growth under multiplication. We analyzed the correctness and security of the construction, showing that its confidentiality reduces to the hardness of integer factorization. Furthermore, our performance evaluation demonstrates that the scheme achieves significant improvements compared to established lattice-based approaches, and noise is efficiently managed without expensive bootstrapping. The main limitation is that the scheme is symmetric and relies on the classical hardness of factoring, which leaves potential vulnerability in the post-quantum setting. As future work, we plan to investigate post-quantum variants of the interposition mechanism.

\bibliographystyle{IEEEtran}	
\bibliography{bibliography}

\end{document}